\documentclass[submission,copyright,creativecommons]{eptcs}
\providecommand{\event}{HAS 2013} 
\usepackage{breakurl}             
\usepackage{xspace}
\usepackage{amsmath,amssymb,amsthm}
\usepackage{paralist}
\usepackage{multicol}
\usepackage{eucal}
\usepackage{tikz}
\usepackage[strings]{underscore}

\newcommand{\hyltl}{\textsf{HyLTL}\xspace}
\newcommand{\ltl}{\textsf{LTL}\xspace}
\newcommand{\X}{\ensuremath{\mathbin{\mathbf{X}}}\xspace}

\newcommand{\U}{\ensuremath{\mathbin{\mathbf{U}}}\xspace}
\newcommand{\R}{\ensuremath{\mathbin{\mathbf{R}}}\xspace}
\newcommand{\F}{\ensuremath{\mathbin{\mathbf{F}}}\xspace}
\newcommand{\G}{\ensuremath{\mathbin{\mathbf{G}}}\xspace}
\newcommand{\fcs}{\ensuremath{CS}\xspace}
\newcommand{\mmodels}{\Vdash}
\newcommand{\cmodels}{\vdash}
\DeclareMathOperator{\cl}{cl}
\DeclareMathOperator{\cs}{mc}

\newcommand{\on}{\text{\it on}}
\newcommand{\off}{\text{\it off}}

\newcommand{\op}{\ensuremath{OP}\xspace}

\newtheorem{definition}{Definition}
\newtheorem{theorem}{Theorem}

\newcommand{\dims}[1]{k}

\DeclareMathOperator{\dom}{\mathrm{dom}}

\newcommand{\bbR}{\mathbb{R}}

\newcommand{\cvF}{\mathcal{F}}

\newcommand{\cvV}{\mathcal{V}}

\newcommand{\autH}{\mathcal{H}}

\newcommand{\Loc}{\mathrm{Loc}}
\newcommand{\Edg}{\mathrm{Edg}}
\DeclareMathOperator{\Dyn}{{Dyn}}

\DeclareMathOperator{\Res}{{Res}}
\DeclareMathOperator{\Init}{{Init}}

\DeclareMathOperator{\Val}{{Val}}

\newcommand{\fstate}{\mathit{fstate}}
\newcommand{\lstate}{\mathit{lstate}}

\newcommand{\ltime}{\mathit{ltime}}

\newcommand{\bx}{\mathbf{x}}
\newcommand{\by}{\mathbf{y}}

\newcommand{\bz}{\mathbf{z}}
\newcommand{\trans}[1]{\xrightarrow{#1}}

\DeclareMathOperator{\trajs}{\mathit{Trajs}}

\newcommand{\rest}{\mathop{\downarrow}}

\newcommand{\valunion}{\mathrel{\sqcup}}

\newcommand{\cut}[1]{ }

\renewcommand{\quote}[1]{\emph{``#1''}}

\title{\hyltl: a temporal logic for model checking hybrid systems}
\author{Davide Bresolin
\institute{Universit\`{a} degli Studi di Verona \\ Dipartimento di Informatica \\ Verona, Italy}
\email{davide.bresolin@univr.it}
}
\def\titlerunning{\hyltl: a temporal logic for model checking hybrid automata}
\def\authorrunning{D. Bresolin}
\begin{document}
\maketitle

\begin{abstract}
The model-checking problem for hybrid systems is a well known challenge in the scientific community. Most of the existing approaches and tools are limited to safety properties only, or operates by transforming the hybrid system to be verified into a discrete one, thus loosing information on the continuous dynamics of the system.
In this paper we present a logic for specifying complex properties of hybrid systems called \hyltl, and we show how it is possible to solve the model checking problem by translating the formula into an equivalent hybrid automaton. In this way the problem is reduced to a reachability problem on hybrid automata that can be solved by using existing tools.
\end{abstract}

\section{Introduction}

\emph{Hybrid systems} exhibit both a discrete and a continuous behaviour with a tight interaction between the two. Typical examples include discrete controllers that operate in a continuous environment, like automotive power-train systems, where a four stroke engine is modeled by a switching continuous system and is controlled by a digital controller.
In order to model and specify hybrid systems in a formal way, the notion of \emph{hybrid automata} has been introduced~\cite{Alur,Hen00a,maler91from}. Intuitively, a hybrid automaton is a ``finite-state automaton'' with continuous variables that evolve according to dynamics characterizing each discrete state (called a {\em location} or {\em mode}). 
Of particular importance in the analysis of hybrid automata is the \emph{model checking problem}, that is, the problem of verifying whether a given hybrid automaton respects some property of interest. 
The state of a hybrid automaton consists of the pairing of a discrete location with a vector of continuous variables, therefore it has the cardinality of continuum. 
This makes the model checking problem computationally difficult. Indeed, even for very simple properties and systems, this problem is not decidable~\cite{henzinger98whats}.

Many different approaches have been used in the literature to tackle this problems.
For simple classes of hybrid systems, like timed automata, the model checking problem can be solved exactly~\cite{timed_automata}, and tools like Kronos \cite{kronos97} and \textsc{UPPAAL}~\cite{uppaal} can be used to verify CTL properties of timed automata. For more complex classes of systems, the problem became undecidable. Nevertheless, many different approximation techniques may be used to obtain an answer, at least in some cases.
Tools like PhaVer \cite{Frehse2008} and SpaceEx~\cite{Frehse2011}  can compute approximations of the reachable set of hybrid automata with linear dynamics, and thus can be used to verify safety properties.
Other tools, like HSOLVER~\cite{hsolver}, and Ariadne~\cite{ijrnc2012}, can manage also systems with nonlinear dynamics, but are still limited to safety properties only.

We are aware of only very few approaches that can specify and verify complex properties of hybrid systems in a systematic way. A first attempt was made in~\cite{Lamport93}, where an extension of the Temporal Logic of Actions called TLA+ is used to specify and implement the well-known gas burner example. Later on, Signal Temporal Logic (STL), an extension of the well-known Metric Interval Logic to hybrid traces, has been introduced to monitor hybrid and continuous systems~\cite{Maler2004}. More recent approaches include the tool KeYmaera~\cite{Platzer2008}, that uses automated theorem proving techniques to verify nonlinear hybrid systems symbolically, and the logic HRELTL~\cite{Cimatti09}, that is supported by an extension of the discrete model checker NuSMV, but it is limited to systems with linear dynamics. 

In this paper we present an alternative approach for model checking hybrid systems. We define a logic for specifying  properties called \hyltl, that extends the well-known \ltl logic to hybrid traces. Then, we show how it is possible to translate a formula of this logic into an equivalent hybrid automaton. In this way the model checking problem is reduced to a reachability problem on the composition of the automaton representing the system and the one representing the (negation of the) formula, and can be solved by using existing tools. An example of verification of a simple property using PhaVer is given to show the feasibility of the approach.

\subsection{Preliminaries}

Before formally defining hybrid automata and the syntax and semantics of \hyltl we need to introduce some basic terminology.
Throughout the paper we fix the \emph{time axis} to be the set of non-negative real numbers $\bbR^+$. An \emph{interval} $I$ is any convex subset of $\bbR^+$, usually denoted as $[t_1, t_2] = \{t \in \bbR^+ : t_1 \leq t \leq t_2\}$. 

We also fix a countable universal set $\cvV$ of \emph{variables}, ranging over the reals. Given a set of variables $X \subseteq \cvV$, a \emph{valuation} over $X$ is a function $\bx: X \mapsto \bbR$ that associates a value to every variable in $X$. The set $\Val(X)$ is the set of all valuations over $X$. Given a valuation $\bx$ and a subset of variables $Y \subseteq X$, we denote the \emph{restriction} of $\bx$ to $Y$ as $\bx\rest Y$. 
The restriction operator is extended to sets of valuations in the usual way. A valuation $\bx$ over $X$ and a valuation $\by$ over $Y$ \emph{agree} when they assign the same value to common variables, i.e., $\bx\rest X\cap Y = \by \rest X \cap Y$. When valuations $\bx$ over $X$ and $\by$ over $Y$ agree, we denote by $\bx\valunion\by$ the \emph{union} of $\bx$ and $\by$, defined as the valuation $\bz$ such that $\bz\rest X = \bx$ and $\bz\rest Y = \by$. Notice that valuations over disjoint sets of variables always agree, and thus their union is always defined.

\medskip

A notion that will play an important role in the paper is the one of \emph{trajectory}. A trajectory over a set of variables $X$ is a \emph{differentiable} function $\tau: I \mapsto \Val(X)$, where $I$ is a left-closed interval with left endpoint equal to $0$. Since $\tau$ is differentiable, its derivative is defined in every point of the domain but the endpoints: we denote with $\dot{\tau}$ the corresponding function giving the value of the derivative of $\tau$ for every point in the interior of $I$ (note that $\dot\tau$ might not be differentiable neither continuous).
With $\dom(\tau)$ we denote the domain of $\tau$, while with $\tau.\ltime$ (the \emph{limit time} of $\tau$) we define the supremum of $\dom(\tau)$. The \emph{first state} of a trajectory is $\tau.\fstate = \tau(0)$, while, when $\dom(\tau)$ is right-closed, the \emph{last state} of a trajectory is defined as $\tau.\lstate = \tau(\tau.\ltime)$. We denote with $\trajs(X)$ the set of all trajectories over $X$. Given a subset $Y \subseteq X$, the \emph{restriction} of $\tau$ to $Y$ is denoted as $\tau\rest Y$ and it is defined as the trajectory $\tau' : \dom(\tau) \mapsto \Val(Y)$ such that $\tau'(t) = \tau(t)\rest Y$ for every $t \in \dom(\tau)$. 


\medskip

Variables will be used in the paper also to build \emph{constraints}: conditions on the value of variables and on their derivative that can define sets of valuations, sets of trajectories, and jump relations. Formally, given a set of variables $X$, and a set of mathematical operators $\op$ (e.g. $+$, $-$, $\cdot$, exponentiation, $\sin$, $\cos$, \dots), we define the corresponded set of \emph{dotted variables} $\dot{X}$ as $\{\dot{x} | x \in X\}$ and the set of \emph{primed variables} $X'$ as $\{x' | x \in X\}$. We use $\op$, $X$, $\dot{X}$ and $X'$ to define the following two classes of constraints.

\begin{itemize}

	\item \emph{Jump constraints}: expressions built up from variables in $X \cup X'$, constants from $\bbR$, mathematical operators from $\op$ and the usual equality and inequality relations ($\leq$, $=$, $>$, \dots). Examples of jump constraints are $x' = 4 y + z$, $x^2 \leq y'$, $y' > \cos(y)$.

	\item \emph{Flow constraints}: expressions built up from variables in $X \cup \dot{X}$, constants from $\bbR$, mathematical operators from $\op$ and the usual equality and inequality relations ($\leq$, $=$, $>$, \dots). Examples of flow constraints are $\dot{x} = 4 y + z$, $x' + y \geq 0$, $\sin(x) > \cos(y')$.
\end{itemize}



We use jump constraints to give conditions on pairs of valuations $(\bx, \bx')$. Given a jump constraint $c$, we say that $(\bx, \bx')$ respects $c$, and we denote it with $(\bx, \bx') \cmodels c$, when, by replacing every variable $x$ with its value in $\bx$ and every primed variable $x'$ with the value of the unprimed variable in $\bx'$ we obtain a solution for $c$.

Flow constraints will be used to give conditions on trajectories. Given a flow constraint $c$, we say that a trajectory $\tau$ respects $c$, and we denote it with $\tau \cmodels c$, if and only if for every time instant $t \in \dom(\tau)$, both the value of the trajectory $\tau(t)$ and the value of its derivative $\dot{\tau}(t)$ respect $c$ (we assume that $\dot{\tau}(t)$ always respects $c$ for $t=0$ and $t=\tau.\ltime$).

\medskip

To conclude the preliminary section, we recall from the introduction that hybrid systems interleaves continuous and discrete evolution. For this reason their behaviors are usually defined in therms of \emph{hybrid traces}, mixing continuous trajectories with discrete events. Formally, given a (possibly finite) set of \emph{actions} (or discrete events) $A$ and a (possibly finite) set of variables $X$, an \emph{hybrid trace over $A$ and $X$} is any infinite sequence $\alpha = \tau_1 a_1 \tau_2 a_2 \tau_3 a_3 \ldots$ such that $\tau_i$ is a trajectory over $X$ and $a_i$ is an action in $A$ for every $i \geq 1$. 

Notice that this definition of hybrid traces allows an infinite sequence of discrete events to occur in a finite amount of time (Zeno behaviors). This will not be a problem for the semantics of the logic nor for the correctness of the model checking algorithm. Zeno behaviors are usually not desirable in real applications, but are very difficult to exclude completely from the language and the formal model of the system. For this reason we choose to include Zeno behaviors in the semantics of our logic at a first stage, leaving a more comprehensive study of this aspect as future work.

\section{Hybrid Automata}\label{sec:ha}

An hybrid automaton is a finite state machine enriched with continuous dynamics labelling each discrete state (or \emph{location}), that alternates continuous and discrete evolution. In continuous evolution, the discrete state does not change, while time passes and the evolution of the continuous state variables follows the dynamic law associated to the current location. A discrete evolution step consists of the activation of a \emph{discrete transition} that can change both the current location and the value of the state variables, in accordance with the reset function associated to the transition. 

The definition of hybrid automata given in this paper extends the one given in~\cite{Alur} to support the parallel composition of automata. Usually, compositional formalisms introduce either a distinction between internal variables and actions (that are hidden to the other automata) and external ones, that are shared between automata (possibly with a further partition between inputs and outputs)~\cite{Lynch03}. This is usually justified by the need to describe in a precise way how a component interacts with the environment and the other components of the system. However, in this paper we are more interested in the analysis of a single component on its own, and this distinction between ``local'' and ``global'' variables and actions is not strictly necessary. Hence, to simplify the definitions and the proofs we choose to consider all variables and actions as shared between the different automata. 

\begin{definition}\label{def:ha-syntax}
   A \emph{hybrid automaton} is a tuple $\autH=\langle\Loc,X,A,\Edg,\Dyn,\Res,\Init\rangle$ such that:
   
   \begin{compactenum}
   	\item $\Loc$ is a finite set of \emph{locations};
   	\item $X$ is a finite set of \emph{variables};
		\item $A$ is a finite set of \emph{actions};
		\item $\Edg \subseteq \Loc \times A \times \Loc$ is a set of \emph{discrete transitions};
		\item $\Dyn$ is a mapping that associates to every location $\ell \in \Loc$ a set of flow constraints $\Dyn(\ell)$ over $X \cup \dot{X}$ describing the \emph{dynamics} of $\ell$;
		\item $\Res$ is a mapping that associates every discrete transition $(\ell,e,\ell') \in \Edg$ with a set of jump constraints $\Res(\ell,e,\ell')$ over $X \cup X'$ describing the guard and reset function of the transition;
		\item $\Init \subseteq \Loc$ is a set of \emph{initial locations}.
\end{compactenum}
\end{definition}

Composition is defined as a binary operation on hybrid automata. 

\begin{definition}\label{def:ha-composition}
	Given two hybrid automata $\autH_1 = \langle\Loc_2,X_1,A_1,\Edg_1,\Dyn_1,\Res_1,\Init_1\rangle$ and \linebreak $\autH_2 = \langle\Loc_2,X_2,A_2,\Edg_2,\Dyn_2,\Res_2,\Init_2\rangle$, we define their composition $\autH_1 \| \autH_2$ as the hybrid automaton $\autH = \langle\Loc,X,A,\Edg,\Dyn,\Res,\Init\rangle$ such that:
\begin{compactenum}
	\item $\Loc = \Loc_1 \times \Loc_2$;

	\item $X = X_1 \cup X_2$;

	\item $A = A_1 \cup A_2$;

	\item $((\ell_1,\ell_2),a,(\ell_1',\ell_2')) \in \Edg$ iff
		\begin{inparaenum}[\it (i)]
			\item $a \in A_i$ and $(\ell_i, a, \ell_i') \in \Edg_i$, or
			\item $a \not\in A_i$ and $\ell_i = \ell_i'$,
		\end{inparaenum}
		for $i = 1,2$;

	\item 	for every $(\ell_1,\ell_2) \in \Loc$ we have that $\Dyn((\ell_1,\ell_2)) = \Dyn_1(\ell_1) \cup \Dyn_2(\ell_2)$;

	\item for every $((\ell_1,\ell_2),a,(\ell_1',\ell_2'))\in\Edg$, we have that $\Res((\ell_1,\ell_2),a,(\ell_1',\ell_2'))$ is 
		the minimal set of jump constraints such that, for $i = 1,2$:
		\begin{compactenum}[\it (i)]
			\item if $a \in A_i$ then $\Res_i(\ell_i,a,\ell_i')\subseteq\Res((\ell_1,\ell_2),a,(\ell_1',\ell_2'))$, and
			\item if $a \not\in A_i$ then $\{x = x' | x \in X_i \setminus X_{3-i}\}\subseteq\Res((\ell_1,\ell_2),a,(\ell_1',\ell_2'))$;
		\end{compactenum}
	\item $\Init = \Init_1 \times \Init_2$.
\end{compactenum}
\end{definition}

The \emph{state} of an hybrid automaton $\autH$ is a pair $(\ell,\bx)$, where $\ell\in \Loc$
is a location and $\bx \in \Val(X)$ is a valuation for the continuous variables.
A state $(\ell,\bx)$ is said to be \emph{admissible} if $(\ell,\bx) \cmodels \Dyn(\ell)$.
Transitions can be either \emph{continuous}, capturing the continuous evolution of the state, or 
\emph{discrete}, capturing instantaneous and discontinuous changes of the state.
Formally, they are defined as follows.

\begin{definition}\label{def:ha-semantics}
Let $\autH$ be a hybrid automaton.
The \emph{continuous transition relation} $\trans{\tau}$ between admissible states, 
where $\tau$ is a bounded trajectory over $X$, is
defined as follows:
\begin{equation}\label{eq:cont-trans}
(\ell,\bx) \trans{\tau} (\ell,\bx')  \iff 
  \tau.\fstate=\bx,\ \tau.\lstate=\bx', \text{ and } \tau \cmodels \Dyn(\ell).
\end{equation}
%
%
The \emph{discrete transition relation} $\trans{e}$ between admissible states, where
$a \in A$, is defined as follows:
\begin{equation}\label{eq:disc-trans}
(\ell,\bx) \trans{a} (\ell',\bx') \iff 
  (\ell,a,\ell') \in \Edg \wedge \bx \cmodels \Dyn(\ell) \wedge \bx' \cmodels \Dyn(\ell') \wedge (\bx,\bx') \cmodels \Res(\ell,a,\ell').
\end{equation}
\end{definition}

This semantics allows us to define the notion of hybrid traces \emph{generated} by the hybrid automaton $\autH$.

\begin{definition}\label{def:ha-trace}
Let $\autH$ be a hybrid automaton, and let $\alpha = \tau_1 a_1 \tau_2 a_2 \ldots$ be a (finite or infinite) hybrid trace over $X$ and $A$. We say that $\alpha$ is \emph{generated} by $\autH$ if there exists a corresponding sequence of locations $\ell_1 \ell_2 \ldots$ such that $\ell_1 \in \Init$ and, for every $i \geq 1$:
\begin{inparaenum}[(i)]
\item  $(\ell_i,\tau_i.\fstate) \trans{\tau_i} (\ell_i,\tau_i.\lstate)$, and
\item $(\ell_i,\tau_i.\lstate) \trans{a_i} (\ell_{i+1},\tau_{i+1}.\fstate)$.
\end{inparaenum}
\end{definition}

Figure~\ref{fig:thermostat} depicts an example of a simple hybrid automaton that models a thermostat. The variable $x$ represents the temperature of the room. In the location \emph{idle} the heater is off and temperature decreases according to the flow condition $\dot{x} = -0.2 x$. The automaton is allowed to stay in \emph{idle}  while the temperature is equal or greater to $17$, and can jump to \emph{heat} as soon as the temperature decrease under $19$ degrees. In location \emph{heat} the heater is turned on and the temperature increases according to the flow condition  $\dot{x} = 30 - 0.2x$. The automaton is allowed to stay in \emph{heat} until the temperature is less or equal to $23$, and can jump to \emph{idle} when it is greater than $21$. Initially, the heater is off.

\begin{figure}[htbp]
\begin{center}
	\begin{tikzpicture}[scale=1]
		 \begin{scope}[thick]
			\draw (0,0) node[draw,rectangle,line width=3pt](v1) {$\begin{array}{c}
					 \mathit{idle} \\ \dot{x} = -0.2 x \\ x \geq 17
				 \end{array}$};
			\draw (5,0) node[draw,rectangle](v2) {$\begin{array}{c}
					 \mathit{heat} \\ \dot{x} = 30 -0.2 x \\ x \leq 23
				 \end{array}$};
			\draw [->] (v1) .. controls +(+1.4,+1.4) and +(-1.4,1.4) .. (v2)
						node[midway,above] {$\mathit{on}\ x \leq 19\ x'=x $};		
			\draw [->] (v2) .. controls +(-1.4,-1.4) and +(1.4,-1.4) .. (v1)
					  node[midway,below] {$\mathit{off}\ x \geq 21\ x'=x $};
		 \end{scope}
	\end{tikzpicture}
\caption{The thermostat automaton $\autH_T$.}
\label{fig:thermostat}
\end{center}
\end{figure}

\section{\hyltl: syntax and semantics}\label{sec:hyltl}

The logic \hyltl is an extension of the well-known temporal logic \ltl to hybrid traces. Given a \emph{finite} set of actions $A$ and a \emph{finite} set of variables $X$, the language of \hyltl is defined from a set of \emph{flow constraints} \fcs over $X$ by the following grammar:
\begin{equation}\label{eq:grammar}
\begin{split}
\varphi ::= \	& c \in \fcs \mid a \in A \mid  \top \mid \bot \mid \neg \varphi \mid
						\varphi \land \varphi \mid \varphi \lor \varphi \mid
						\X \varphi \mid \varphi \U \varphi \mid \varphi \R \varphi 
\end{split}
\end{equation} 

\noindent In \hyltl flow constraints from \fcs and actions from $A$ take the role of propositional letters in standard temporal logics, $\top$ and $\bot$ are the logical constants \emph{true} and \emph{false}, $\neg$, $\land$ and $\lor$ are the usual boolean connectives, $\X$, $\U$ and $\R$ are hybrid counterpart of the standard \emph{next}, \emph{until} and \emph{release} temporal operators. 

\medskip

Let $\alpha = \tau_1 a_1 \tau_2 a_2 \ldots$ be an infinite hybrid trace. For every $i > 0$, the truth value of a \hyltl formula $\varphi$ over $\alpha$ at position $i$ is given by the truth relation $\mmodels$, formally defined as follows:

\begin{itemize}
	\item for every $c \in \fcs$, $\alpha, i \mmodels c$ iff $\tau_i \cmodels c$;
	\item for every $a \in A$, $\alpha, i \mmodels a$ iff $i > 1$ and $a_{i-1} = a$;
	\item $\alpha,i \mmodels \top$ and $\alpha,i \not\mmodels \bot$;
	\item $\alpha,i \mmodels \neg\varphi$ iff $\alpha,i\not\mmodels\varphi$;
	\item $\alpha,i \mmodels \varphi \land \psi$ iff $\alpha,i \mmodels \varphi$ and $\alpha,i \mmodels \psi$;
	\item $\alpha,i \mmodels \varphi \lor \psi$ iff $\alpha,i \mmodels \varphi$ or $\alpha,i \mmodels \psi$;
	\item $\alpha,i \mmodels \X\varphi$ iff $\alpha,i+1 \mmodels \varphi$;
	\item $\alpha,i \mmodels \varphi \U \psi$ iff there exists $j \geq i$ such that $\alpha,i\mmodels\psi$, and for every $i\leq k < j$, $\alpha,k\mmodels\varphi$;	
	\item $\alpha,i \mmodels \varphi \R \psi$ iff for all $j \geq i$, if for every $i \leq k < j$, $\alpha, k \not\mmodels \varphi$ then $\alpha,j \mmodels \psi$.
\end{itemize}

\noindent Intuitively, the operator $\X$ (``next discrete action'') requires that some formula hold after the next discrete action. The operator $\U$ forces a property to hold until some other property holds.

The next and until operator are sufficient to express other useful temporal operators, namely the ``always'' operator $\G$ and the ``eventually'' operator $\F$. They can be defined as usual:
\begin{eqnarray*}
	\F \varphi = \top \U \varphi & \qquad\qquad & \G \varphi = \neg \F \neg \varphi
\end{eqnarray*}

\noindent Using the base and derived operators of \hyltl it is possible to express a number of different properties, including safety properties, liveness properties, reactivity properties and so on. For instance, in the case of the Thermostat example described in Figure~\ref{fig:thermostat}, the usual property of \quote{keeping the temperature of the room between 15 and 25 degrees} is defined by the following formula:
\begin{equation}
	\varphi_{safe} = \G \left(x \geq 15 \land x \leq 25\right)
\end{equation}

\noindent Safety properties are the most common type of requirements on hybrid systems, but \hyltl can easily go beyond them. For example, it can express reactivity properties like \quote{whenever the heater is turned ON it is eventually turned OFF}:
\begin{equation}
	\varphi_{react} = \G \left( \on \rightarrow \F \off \right)
\end{equation}

\noindent Finally, the presence of discrete actions and of flow constraints in the syntax of \hyltl allow this logic to express properties involving both the discrete and the continuous behaviors of the system, thus capturing their hybrid nature. An example on the thermostat can be \quote{it is not possible that the heater turns on when the temperature is above 21 degrees}:
\begin{equation}
	\varphi_{hyb} = \neg \F \left( x \geq 21 \land \X \on \right)
\end{equation}

\section{From \hyltl to Hybrid Automata with B\"uchi conditions}

In this section we show how to translate a \hyltl formula into an equivalent hybrid automaton, that is, an hybrid automaton representing all traces that satisfy the formula. In this paper we use a declarative construction, similar to the tableau construction presented in~\cite{mcbook} for model checking of \ltl formulas. 

In analogy with the classical construction, we have to enrich the notion of hybrid automata with a suitable \emph{acceptance condition} to identify the traces generated by the automaton that fulfills the semantics of the until operator. To this end, we mark some of the locations of the hybrid automaton as \emph{final location} and we extend Definition~\ref{def:ha-syntax} with a \emph{Generalized B\"uchi accepting condition} for final locations.

\begin{definition}\label{def:ha-buechi}
 A \emph{Hybrid Automaton with Generalized B\"uchi condition (GBHA)} is a tuple $\autH=\langle\Loc,$ $X,A,\Edg,\Dyn,\Res,\Init,\cvF\rangle$ such that $\langle\Loc,X,A,\Edg,\Dyn,\Res,\Init\rangle$ is a Hybrid Automaton, and  \linebreak $\cvF = \{F_1, \ldots, F_n\} \subseteq 2^{\Loc}$ is a finite set of sets of final locations.
\end{definition}

We say that an hybrid trace $\alpha = \tau_1 a_1 \tau_2 a_2 \ldots$ is \emph{accepted} by a GBHA $\autH$ if there exists an \emph{infinite} sequence of locations $\ell_1 \ell_2 \ldots$ such that:
\begin{compactenum}[(i)]
\item $\ell_1 \in \Init$;
\item for every $i \geq 1$, $(\ell_i,\tau_i.\fstate) \trans{\tau_i} (\ell_i,\tau_i.\lstate)$; 
\item for every $i \geq 1$, $(\ell_i,\tau_i.\lstate) \trans{a_i} (\ell_{i+1},\tau_{i+1}.\fstate)$;
\item for every $F_j \in \cvF$ there exists $f_j \in F_j$ that occurs infinitely often in the sequence.
\end{compactenum}

\noindent Notice that $\alpha$ must be a sequence generated by $\autH$. However, not all sequences generated by the automaton are accepting: only those that respect the additional accepting condition are considered.

Given a \hyltl formula $\varphi$, we show now how to build a GBHA that accepts exactly all hybrid traces that satisfies $\varphi$. As in the case of \ltl , we assume that the formula is in negated normal form, and we define the notion of \emph{closure} of a \hyltl formula as follows.

\setlength{\multicolsep}{0pt}

\begin{definition}\label{def:closure}
Give an \hyltl formula $\varphi$ and a set of actions $A$ we define the \emph{closure of $\varphi$} as the smallest set $\cl(\varphi)$ of formulas satisfying the following conditions:
\begin{multicols}{2}
\begin{compactitem}
	\item $\varphi \in \cl(\varphi)$;
	\item for each $a \in A$, $a \in \cl(\varphi)$;
	\item if $\psi_1 \in \cl(\varphi)$ then $\neg{\psi_1} \in \cl(\varphi)$;
	\item if $\psi_1 \land \psi_2 \in \cl(\varphi)$ then $\psi_1, \psi_2 \in \cl(\varphi)$;
	\item if $\psi_1 \lor \psi_2 \in \cl(\varphi)$ then $\psi_1, \psi_2 \in \cl(\varphi)$;
	\item if $\X \psi_1 \in \cl(\varphi)$ then $\psi_1 \in \cl(\varphi)$;
	\item if $\psi_1 \U \psi_2 \in \cl(\varphi)$ then $\psi_1, \psi_2 \in \cl(\varphi)$;
	\item if $\psi_1 \R \psi_2 \in \cl(\varphi)$ then $\psi_1, \psi_2 \in \cl(\varphi)$;
\end{compactitem}
\end{multicols}
\end{definition}

\noindent The subsets of $\cl(\varphi)$ will label the locations of the hybrid automaton.  We aim to construct an automaton such that if a location is labelled by a subset $M \subseteq \cl(\varphi)$ then every accepting run starting from that location satisfies all formulas in $M$. For this reason, we can ignore every subset of the closure that is clearly inconsistent or subsumed by a consistent superset of formulas, and concentrate our attention to \emph{maximally consistent} subsets of $\cl(\varphi)$ defined as follows.

\begin{definition}\label{def:maximally-consistent}
A subset $M \subseteq \cl(\varphi)$ is \emph{maximally consistent} if it respects the following conditions:
\begin{multicols}{2}
\begin{compactenum}
	\item $\top \in M$;
	\item $\psi_1 \in M$ iff $\neg{\psi_1} \not\in M$;
	\item $\psi_1 \land \psi_2 \in M$ iff $\psi_1 \in M$ and $\psi_2 \in M$;
	\item $\psi_1 \lor \psi_2 \in M$ iff $\psi_1 \in M$ or $\psi_2 \in M$;
	\item if  $a \in M$ then for each $b \neq a$, $b \not\in M$.
\end{compactenum}
\end{multicols}
\end{definition}

\noindent The first four conditions are the usual boolean consistency conditions, while conditions 5 guarantees that actions in $A$ are mutually exclusive.
Let $\cs(\varphi)$ be the set of maximally consistent subsets of $\cl(\varphi)$. We are going to use $\cs(\varphi)$ as the location of the GBHA representing $\varphi$.

\begin{definition}\label{def:gbha-varphi}
Let $\varphi$ be a \hyltl formula over a set of variables $X$ and actions $A$. The corresponding GBHA $\autH_\varphi = \langle\Loc,X,A,\Edg,\Dyn,\Res,\Init,\cvF\rangle$ is defined as follows:
\begin{compactitem}
	\item $\Loc = \cs(\varphi)$;
	\item for each $M, M' \in \cs(\varphi)$ and $a \in A$, $(M,a,M') \in \Edg$ if and only if:
		\begin{compactitem}
			\item $a \in M'$;
			\item $\X \psi_1 \in M$ iff $\psi_1 \in M'$;
			\item $\psi_1 \U \psi_2 \in M$ iff $\psi_2 \in M$ or ($\psi_1 \in M$ and $\psi_1 \U \psi_2 \in M'$);
			\item $\psi_1 \R \psi_2 \in M$ iff $\psi_1,\psi_2 \in M$ or ($\psi_2\in M$ and $\psi_1 \R \psi_2 \in M'$);
		\end{compactitem}
	\item for each $M \in \cs(\varphi)$, $\Dyn(M) = \{c \in \fcs | c \in M\}$ (i.e., the set of flow constraints of $M$);
	\item for each $(M, a, M') \in \Edg$, $\Res(M, a, M') = \top$;
	\item $\Init = \{M \in \cs(\varphi) | \varphi \in M$ and $M \cap A = \emptyset\}$;
	\item for each $\psi_1 \U \psi_2 \in \cl(\varphi)$, $\{M \in \cs(\varphi) | \psi_2 \in M$ or $\neg(\psi_1 \U \psi_2) \in M\} \in \cvF$.
\end{compactitem}
\end{definition}

The conditions on the set of discrete transitions $\Edg$ guarantee that any run of the automaton does not violate the semantics of the temporal operators. The accepting condition $\cvF$ is a set of sets of locations, where every set in $\cvF$ corresponds to an until-formula in $\cl(\varphi)$. They guarantee that whenever a formula $\psi_1 \U \psi_2$ become true in some location, then the promise $\psi_2$ is eventually fulfilled later on in the run.
The following theorem proves that $\autH_\varphi$ accepts exactly all the hybrid traces satisfying $\varphi$.

\begin{theorem}\label{teo:main}
Let $\varphi$ be \hyltl-formula, and $\alpha = \tau_1 a_1 \tau_2 a_2 \ldots$ be an infinite hybrid trace over $X$ and $A$. Then $\alpha,1 \mmodels \varphi$ if and only if $\autH_\varphi$ accepts $\alpha$.
\end{theorem}

\begin{proof}
Suppose that $\alpha, 1 \mmodels \varphi$. For every $i \geq 1$, let $M_i = \{\psi\in \cl(\varphi) | \alpha,i \mmodels \psi\}$. It is easy to see that every $M_i$ is a maximal consistent set of formulas. Moreover, by the definition of $\autH_\varphi$, we have that $M_1 \in \Init$ and that, for every $i \geq 1$, $\tau_i \cmodels \Dyn(M_i)$,
$(M_i, a_i, M_{i+1}) \in \Edg$, and $(\tau_i.\lstate,\tau_{i+1}.\fstate) \cmodels \Res(M_i, a_i, M_{i+1})$.

Hence, $\alpha$ is generated by $\autH_\varphi$. To show that $\alpha$ is also accepted by the automaton, let $\psi_1 \U \psi_2$ be an until formula in $\cl(\varphi)$. Two cases may arise: either there exists some index $j$ such that for all $k \geq j$, $\alpha,k\not\mmodels \psi_1 \U \psi_2$; or there exists an infinite set of indexes $j_1, j_2, \ldots$ such that $\alpha,{j_k}\mmodels \psi_1 \U \psi_2$. In the former case, all maximal consistent sets $M_k$ are such that $\neg(\psi_1 \U \psi_2) \in M_k$. In the latter case, by the semantics of \hyltl, for each $j_k$ there must exists an index $h_k$ such that $\alpha,{h_k}\mmodels \psi_2$: this implies that $\psi_2\in M_{h_k}$. In both cases there exists a final location $f \in \{M \in \cs(\varphi) | \psi_2 \in M$ or $\neg(\psi_1 \U \psi_2) \in M\}$ that occurs infinitely often in the run of the automaton. This proves that $\alpha$ is accepted by $\autH$.

\medskip

Conversely, suppose that $\alpha$ is accepted by $\autH_\varphi$, and let $M_1 M_2 \ldots$ be an accepting sequence of locations for $\alpha$. We prove the claim by showing that the following stronger property holds:
\begin{equation}\label{eq:strongerprop}
\text{\it for every $i \geq 1$ and for every $\psi\in\cl(\varphi)$, $\psi \in M_i$ if and only if $\alpha,i \mmodels \psi$.}
\end{equation}
 
\noindent We proceed by induction on the structure of $\psi$.

\begin{itemize}
	\item $\varphi = c$, with $c = \top$ or $c \in \fcs$. By the definition of $\autH_\varphi$, $c \in \Dyn(M_i)$ and thus, by the definition of the continuous transition relation, $\tau_i \cmodels c$.
	
	\item $\varphi = a$, for some $a \in A$. By the definition of $\autH_\varphi$, $M_i \not\in \Init$: this implies that $i > 1$. Consider now the action $a_{i-1}$ in $\alpha$: by the definition of $\autH_\varphi$, we have that $a_{i-1} = a$ and thus $\alpha,i\mmodels a$.
	
	\item $\varphi = \psi_1 \land \psi_2$. By the definition of consistent set, $\psi_1 \land \psi_2 \in M_i$ if and only if $\psi_1 \in M_i$ and $\psi_2 \in M_i$. By the induction hypothesis, we have that $\alpha,i \mmodels \psi_1$ and $\alpha,i \mmodels \psi_2$. By the semantics of \hyltl, we have that $\alpha,i \mmodels \psi_1 \land \psi_2$.
	
	\item $\varphi = \psi_1 \lor \psi_2$ or $\varphi = \neg\psi_1$. This case is similar to the previous one and thus skipped.
	
	\item $\varphi = \X\psi_1$. By the definition of $\autH_\varphi$, since $(M_i, a_i, M_{i+1}) \in \Edg$ we have that $\psi_1\in M_{i+1}$. By inductive hypothesis, we have that $\alpha,{i+1} \mmodels \psi_1$ and thus that $\alpha,i \mmodels \X\psi_1$.

		
	\item $\varphi = \psi_1\U\psi_2$. Suppose that $\psi_1\U\psi_2 \in M_i$. Then, by the definition of $\autH_\varphi$, either $\psi_1 \in M_i$ or $\psi_2 \in M_i$. By inductive hypothesis, this implies that $\alpha,i \mmodels \psi_1$ or $\alpha,i \mmodels \psi_2$. In the latter case we can conclude that $\alpha,i\mmodels \psi_1\U\psi_2$.
	
	In the former case, by the definition of the discrete transitions $\Edg$, we have $\psi_1\U\psi_2 \in M_{i+1}$. By the definition of the acceptance condition of $\autH_\varphi$, we have that there must exists at least one index $j\geq i+1$ such that $\psi_2 \in M_j$. Let $k$ be the smallest of those indexes: by induction hypothesis we have that $\alpha,{k}\mmodels\psi_2$. We can also prove by induction on $m$ that, for every $i+1 \leq m < k$, $\psi_1, \psi_1\U\psi_2 \in  M_m$. When $m = i+1$, the proof is trivial. When $i+1 < m < k$ we have, by induction hypothesis, that $\psi_1\U\psi_2 \in M_{m-1}$. Since $\psi_2\not\in M_{m-1}$, by the definition of $\Edg$ we have that $\psi_1\U\psi_2\in M_m$ and thus that $\psi_1 \in M_m$. This implies that, by inductive hypothesis on $\varphi$, that $\alpha,m \mmodels \psi_1$, and thus that $\alpha,i \mmodels \psi_1\U\psi_2$.

	Now, suppose that $\alpha,i \mmodels \psi_1 \U \psi_2$. By the semantics of \hyltl, two cases may arise:
	\begin{compactitem}
		\item $\alpha,i \mmodels \psi_2$. By inductive hypothesis we have that $\psi_2 \in M_i$. By the definition of $\autH_\varphi$, this implies that $\psi_1\U\psi_2 \in M_i$.
		\item $\alpha,i \mmodels \psi_1$ and $\alpha,{i+1}\mmodels\psi_1\U\psi_2$. By the semantics of \hyltl, there must exists at least one index $j\geq i+1$ such that $\psi_2 \in M_j$. Let $k$ be the smallest of those indexes: by proceeding as above we can conclude that $\psi_1\U\psi_2 \in M_i$.
	\end{compactitem}
	
	\item Finally, the case when $\varphi = \psi_1\R\psi_2$ is similar to the previous one and can be skipped.
\end{itemize}

\noindent This concludes the proof of property~\eqref{eq:strongerprop}. By definition of $\Init$, $\varphi\in M_1$, and thus, by~\eqref{eq:strongerprop}, we have proved that $\alpha \mmodels \varphi$, as required.
\end{proof}

\section{Model checking \hyltl}

In the previous section we have shown how to build an GBHA that is equivalent to a \hyltl formula. In this section we show how this can be exploited to solve the model checking problem for \hyltl. 

Let $\autH_S$ be a hybrid automaton representing the system under verification, and let $\varphi$ be the \hyltl formula representing a property that the system should respect. Consider the GBHA $\autH_{\neg\varphi}$ that is equivalent to \emph{the negation of the property}: by Theorem~\ref{teo:main}, it accepts all the hybrid traces that \emph{violates} the property we want to verify. Now, if we compose the automaton for the system with the automaton for $\neg\varphi$ we obtain a GBHA $\autH_S \| \autH_{\neg\varphi}$ such that:
\begin{itemize}
	\item it generates only hybrid traces that are generated by $\autH_S$;
	\item it accepts only hybrid traces that are accepted by $\autH_{\neg\varphi}$.
\end{itemize}

\noindent Hence, $\autH_S \| \autH_{\neg\varphi}$ accepts only those hybrid traces that are generated by the system and violates the property. This means that $\autH_S$ respects the property $\varphi$  if and only $\autH_S \| \autH_{\neg\varphi}$ \emph{does not accept any hybrid trace}.

This last property of $\autH_S \| \autH_{\neg\varphi}$ is a reachability property that can be tested by existing tools for the reachability analysis of hybrid automata, and thus allows to use existing technology to verify properties of \hyltl. The only thing that one need to do is to write a procedure implementing the construction of Definition~\ref{def:gbha-varphi} to build the automaton for the negation of the formula, and then send the results to the reachability analysis tool.

Feasibility of this approach has been tested by verifying the Thermostat example of Section~\ref{sec:ha} against the example formula $\varphi_{hyb} = \neg \F \left( x \geq 21 \land \X \on\right)$ of Section~\ref{sec:hyltl}, using the well-known software package PhaVer~\cite{Frehse2008}. 
As a first step, it is necessary to build the automaton for $\neg\varphi_{hyb} = \F \left( x \geq 21\land \X\on \right) = \top \U \left( x \geq 21 \land \X \on \right)$. The application of Definition~\ref{def:gbha-varphi} leads to the GBHA depicted in Figure~\ref{fig:gbha-neg}, where initial locations have a thick border and final locations are grayed out. In this particular case there is only one until formula in the closure, so $\cvF$ is a singleton set.
To simplify the picture, some of the formulas labeling the locations have been left out.

\begin{figure}[htbp]
\resizebox{\textwidth}{!}{\input{newHphi}}
\caption{The GBHA for $\neg\varphi_{hyb}$.}
\label{fig:gbha-neg}
\end{figure}

Since PhaVer allows the composition of hybrid automata, it was not necessary to compute the parallel composition of $\autH_T$ and $\autH_{\neg\varphi_{hyb}}$. However, PhaVer does not allow to specify acceptance conditions on hybrid automata. Hence, to search for an accepting run we have to modify the automaton for the formula as follows:
\begin{itemize}
	\item we assigned to every final state a unique numerical value different from $0$;
	\item we added two auxiliary variables, $f$ and $y$, with initial value $0$;
	\item the derivative of $f$ and $y$ are equal to $0$ in all locations;
	\item when a transition exiting a final state is taken, and the current value of $f$ is $0$,
		$f$ can be non-deterministically reset to the value associated to the final state. At the same time, the current value of $x$ is stored in $y$;
	\item all other transitions leave the values of $f$ and $x$ unchanged.
\end{itemize}

\noindent In this way we exploit the non-determinism to guess the occurrence of a final location that will occur infinitely often in the accepting run of the automaton, and we can check its existence within PhaVer: it is sufficient to compute the set of reachable states of $\autH_T \| \autH_{\neg\varphi_{hyb}}$ and check whether there exists a state where the location of $\autH_{\neg\varphi_{hyb}}$ is a final one, the value of $f$ is equal to the numerical value of that location, and $y = x$. If this is not the case, then there are no accepting runs and 
the property is verified. If, on the other hand, such a state can be found, it may be the case that a loop from a final location to a final location can be built. Notice that in the latter case we cannot conclude that the system falsify the required property: PhaVer computes an \emph{over-approximation} of the exact reachable set of the system. Hence, it may be the case that no accepting run exists, but the test found a loop from a final location to a final due to over-approximation errors.

In the example given in this paper, PhaVer was able to correctly verify that the thermostat respect the property $\varphi_{hyb}$. In the case of this very simple example, the computation time was almost instantaneous (0.2 seconds on an Intel Core Duo 2.0 GHz notebook).

\section{Conclusion}

In this paper we presented \hyltl, a logic that is able to express properties of hybrid traces, and we have shown how to translate a formula of \hyltl to an equivalent hybrid automaton with B\"uchi acceptance conditions.
In this way it is possible to solve the model checking for \hyltl by reducing it to a reachability problem on the composition of the hybrid automata representing the system with the one representing the (negation of the) formula.
Feasibility of the approach has been tested by verifying a very simple example using the well-known tool PhaVer.

This is still a preliminary work, that can be extended in many directions. The expressivity of the logic can be extended by adding jump predicates to the language, to express properties on the reset functions of the system. 
The algorithm for computing the equivalent automaton can be improved by using an on-the-fly approach like the one used in~\cite{Gerth95} for \ltl. Finally, tool support for the logic can be extended by customizing existing reachability analysis tools to manage final locations directly, without the need to introduce additional variables to the model.

\bibliographystyle{eptcs}
\bibliography{has2013}
\end{document}